\documentclass[11pt]{article}
\usepackage{brainberry}
\usepackage{booktabs}
\makeatletter
\newcommand{\restatednote}{\ifthmt@thisistheone\else, \textcolor{black}{Restated}\fi}
\makeatother
\SetKwIF{If}{ElseIf}{Else}{If}{then}{Else if}{Else}{end if}
\SetKwFor{For}{For}{do}{end for}
\SetKw{Return}{Return}
\SetKw{KwRet}{Return}

\title{\resizebox{\textwidth}{!}{Optimal Simulated Annealing for Partition Function Estimation}}
{\author{Author(s)}}{
\author{Heng Guo\thanks{\fontsize{8.5pt}{10pt}\selectfont School of Informatics, University of Edinburgh, Edinburgh, United Kingdom. Email: \texttt{hguo@inf.ed.ac.uk}}
\and Hongyang Liu\thanks{\fontsize{8.5pt}{10pt}\selectfont School of Computer Science, State Key Laboratory for Novel Software
Technology, New Cornerstone Science Laboratory, Nanjing University,
Nanjing, China. Emails:
\texttt{liuhongyang@smail.nju.edu.cn, yinyt@nju.edu.cn, zhangyiyao@smail.nju.edu.cn}}
\and  Xiongxin Yang\thanks{\fontsize{8.5pt}{10pt}\selectfont Department of Computer Science, University of California, Santa Barbara, United States. Email: \texttt{xiongxinyang@ucsb.edu}}
\and  Yitong Yin\footnotemark[2]
\and  Yiyao Zhang\footnotemark[2]}
}

\begin{document}

\maketitle

\begin{abstract}

  In this note, we give a simple analysis of a non-adaptive simulated annealing algorithm for estimating the partition function of Gibbs distributions.
  This yields the most efficient reduction of this kind so far.
  We also establish lower bounds for both general and non-adaptive algorithms, showing that our algorithm is optimal over a broad range of parameters.

\end{abstract}

\section{Introduction}

Counting and sampling are two fundamental computational tasks. 
In a seminal work, Jerrum, Valiant, and Vazirani~\cite{jerrum_random_1986} showed that, for a broad class of problems, their approximate versions are equivalent up to polynomial time.
Inspired by statistical physics, Gibbs distributions and their associated partition functions provide a particularly expressive framework for studying both tasks.
Such Gibbs families underlie algorithms for approximating the volume of convex bodies~\cite{dyer_random_1989}, the permanent of matrices~\cite{jerrumPolynomialtimeApproximationAlgorithm2004}, the partition function of spin systems~\cite{jerrum_polynomial-time_1993}, and the number of satisfying assignments to a CNF formula~\cite{fengFastSamplingCounting2021}.

Formally, let $\Omega$ be a finite state space and let $H:\Omega\to\mathbb{R}_{\ge0}$ be a Hamiltonian.
Throughout, we assume that $H(\Omega)\subseteq\{0\}\cup[1,\infty)$ and that $H^{-1}(0)$ is nonempty.
For each $\beta\in\mathbb{R} \cup \{-\infty\}$, the \emph{Gibbs distribution} associated with $H$ parameterized by $\beta$ is defined by
\begin{equation*}
    \mu_\beta(X)=\frac{\exp(\beta H(X))}{Z(\beta)},
\qquad
Z(\beta)=\sum_{X\in\Omega}\exp(\beta H(X)),
\end{equation*}
where $Z(\beta)$ is called the \emph{partition function}.
We set $Z(-\infty)=|H^{-1}(0)|$ and take $\mu_{-\infty}$ to be uniform on $H^{-1}(0)$.
The sampling task is to draw a sample from $\mu_\beta$, whereas the
counting task is formulated as estimating the ratio of the partition functions at two prescribed parameters $\beta_{\min}$ and $\beta_{\max}$.

In this paper, we revisit the reduction from sampling to counting. 
We first define the sampling oracle model, which is a standard abstraction for this task.
\begin{definition}[Oracle model]
\label{def:oracle-model}
Let $-\infty\le\beta_{\min}\le\beta_{\max}<+\infty$.
On each query $\beta\in[\beta_{\min},\beta_{\max}]$, the oracle independently draws $X\sim\mu_\beta$ and returns only $H(X)$.
\end{definition}
Note that the oracle does not reveal the sample $X$ itself, but only its Hamiltonian value $H(X)$.
Our goal is to estimate the partition ratio $Q=Z(\beta_{\max})/Z(\beta_{\min})$ using these samples.
This reduction, commonly known as \emph{simulated annealing}, interpolates between the endpoints using intermediate parameters called a \emph{cooling schedule}~\cite{Bezakova08,Stefankovic:JACM09,Huber:Gibbs}.
An algorithm is \emph{non-adaptive} if all query parameters are fixed before any samples are observed, and \emph{adaptive} otherwise.
All queries of a non-adaptive algorithm can therefore be made in a single sampling round.
Our oracle model assumes access to exact samples, whereas applications typically provide only approximate samples.
This assumption entails essentially no loss: for an algorithm making at most $N$ oracle calls, if each call returns a sample within total variation distance $\delta/N$ of the target Gibbs distribution, a coupling argument and a union bound show that replacing exact samples with approximate ones changes the success probability by at most $\delta$.

We measure the efficiency of such an algorithm in terms of sample complexity and the number of sampling rounds. 
The \emph{sample complexity} is the total number of oracle queries.
A \emph{sampling round} consists of oracle queries whose parameters are fixed before the round begins and that can be executed in parallel; 
parameters in later rounds may depend on samples observed in earlier rounds.
Thus, the number of sampling rounds measures how many sequential stages of oracle queries are required.
Unless otherwise specified, $\log$ denotes the natural logarithm.
The input of the algorithm also includes two arbitrary known upper bounds
\[
q \ge \log Q, \qquad h \ge \E[X\sim\mu_{\beta_{\max}}]{H(X)},
\]
and we measure the sample complexity and sampling rounds in terms of these parameters and the desired relative error $\varepsilon$.
Without loss of generality, we assume throughout that $q\ge1$ and $h\ge2$.

Early non-adaptive simulated annealing algorithms~\cite{dyer_random_1989,Bezakova08} use $O(q^2\varepsilon^{-2}\log h)$ samples in a single round.
The first improvement in total sampling work came from adaptive simulated annealing~\cite{Stefankovic:JACM09}, which achieved $O(q\varepsilon^{-2}\operatorname{polylog}(q,h))$ samples.
Later methods based on the Tootsie Pop Algorithm (TPA)~\cite{tpa,TPA:journal} and the paired product estimator~\cite{Huber:Gibbs} culminated in $O(q\varepsilon^{-2}\log h)$ samples~\cite{Kolmogorov:COLT18,harrisParameterEstimationGibbs2024}, but remained highly adaptive.
Recent work~\cite{harrisNearOptimalParallelApproximate2026} obtained $O(q\varepsilon^{-2}\log^2 h)$ samples in one round and $O(q\varepsilon^{-2}\log h)$ samples in three rounds.
Here we obtain the latter sample bound in a single non-adaptive round.

\begin{restatable}[Upper bound\restatednote]{theorem}{upperbound}
    \label{thm:main-result}
    For every $\varepsilon\in(0,1)$, there is a non-adaptive algorithm that estimates $Q$ within relative error $\varepsilon$ with probability at least $3/4$ using $O(q\varepsilon^{-2}\log h)$ oracle calls.
\end{restatable}

\pref{tab:prior-results} compares prior work and our result in terms of sample complexity and the number of sampling rounds. 
Here $\widetilde O(\cdot)$ suppresses polylogarithmic factors.

\begin{table}[ht]
\centering
\small
\caption{Comparison of Sample Bounds and Parallelism}
\label{tab:prior-results}
\setlength{\tabcolsep}{4pt}
\renewcommand{\arraystretch}{1.08}
\begin{tabular}{ccc}
\toprule
Result & Samples & Rounds of sampling \\
\midrule
\cite{dyer_random_1989,Bezakova08} & $O(q^2\varepsilon^{-2}\log h)$ & $1$ \\
\cite{Stefankovic:JACM09} & $O(q\varepsilon^{-2}\operatorname{polylog}(q,h))$ & $\widetilde O(\sqrt q)$ \\
\cite{Huber:Gibbs,Kolmogorov:COLT18,harrisParameterEstimationGibbs2024} & $O(q\varepsilon^{-2}\log h)$ & $\widetilde O(q)$ \\
\cite{harrisNearOptimalParallelApproximate2026} & $O(q\varepsilon^{-2}\log^2 h)$ & $1$ \\
\cite{harrisNearOptimalParallelApproximate2026} & $O(q\varepsilon^{-2}\log h)$ & $3$ \\
This note & $O(q\varepsilon^{-2}\log h)$ & $1$ \\
\bottomrule
\end{tabular}
\end{table}

Combined with the parallel samplers of~\cite{liuParallelizeSingleSiteDynamics2024} and~\cite{chen2025efficient}, 
our reduction also gives $\RNC$ approximate counting algorithms for the corresponding models.
See~\cite{harrisNearOptimalParallelApproximate2026} for more details.

On the lower-bound side, Kolmogorov~\cite{Kolmogorov:COLT18} proved a lower bound of $\Omega(q\varepsilon^{-2})$ oracle calls for partition-ratio estimation under our oracle model.
We obtain two stronger lower bounds: one for general algorithms, allowing arbitrary adaptivity, and one for non-adaptive algorithms.
\begin{restatable}[General lower bound\restatednote]{theorem}{generallowerbound}
    \label{thm:general-lower-bound}
    Fix $q\geq256$, $h\geq\e^{256}$, and $\varepsilon\in(0,1/4)$.
    Suppose an algorithm estimates $Q$ within relative error $\varepsilon$ with probability at least $3/4$ on every input with parameters $q$ and $h$.
    If it uses at most $N$ queries per execution, then $N=\Omega(\min\{q\log h,q^2\}\cdot\varepsilon^{-2})$.
\end{restatable}

For non-adaptive algorithms, we obtain the following stronger bound.
\begin{restatable}[Non-adaptive lower bound\restatednote]{theorem}{nonadaptivelowerbound}
    \label{thm:non-adaptive-lower-bound}
    Fix $q\geq256$, $h\geq\e^{256}$, and $\varepsilon\in(0,1/4)$.
    Suppose a non-adaptive algorithm estimates $Q$ within relative error $\varepsilon$ with probability at least $3/4$ on every input with parameters $q$ and $h$.
    If it uses at most $N$ queries per execution, then $N=\Omega(\min\{q\log h,\e^q\}\cdot\varepsilon^{-2})$.
\end{restatable}

These lower bounds are further refinements of the lower bounds in \cite{Kolmogorov:COLT18} and \cite{Stefankovic:JACM09}, respectively.
They show that our non-adaptive algorithm is optimal when $\log h\leq q$, and optimal among non-adaptive algorithms when $\log h\leq\e^q/q$.
The $\min$ in both \pref{thm:general-lower-bound} and \pref{thm:non-adaptive-lower-bound} are necessary.
When adaptivity is allowed, the TPA algorithm may only need $O(q^2\varepsilon^{-2})$ queries. 
See \cite[Section 2.2, Algorithm 3]{Kolmogorov:COLT18} and the discussion therein.
For the non-adaptive case, 
we may simply query $H(X)$ at $\beta=\beta_{\max}$ and take the empirical mean of $\e^{(\beta_{\min}-\beta_{\max})H(X)}$ to get an $\varepsilon$-approximation of $1/Q$.
The number of queries needed in this case is $O(\e^q\varepsilon^{-2})$.

In most natural applications, $q$ and $h$ are polynomially related.
For example, in the hard-core model on an $n$-vertex graph~\cite[Section~5.1]{harrisNearOptimalParallelApproximate2026}, each independent set $I\in\mathcal I$ has probability proportional to $\lambda^{|I|}$, where $\mathcal I$ is the set of independent sets and $\lambda>0$ is the fugacity.
Let $H(I)=|I|$, $\beta_{\min}=-\infty$, and $\beta_{\max}=\log\lambda$.
Hence we have $Z(\beta_{\min})=1$ and $Q=\sum_{I\in\mathcal I}\lambda^{|I|}\leq(1+\lambda)^n$.
We take $q=n\log(1+\lambda)$ and $h=n$ so that $\log h\leq q$ holds for fixed $\lambda>0$ and sufficiently large $n$.

\section{Proof of the Upper Bound}

In this section, we prove \pref{thm:main-result} by simply combining a non-adaptive cooling schedule introduced in~\cite{harrisNearOptimalParallelApproximate2026} with the classic product estimator.
\upperbound*



Our key idea, inspired by the $k$-SAT counting work of Feng, Guo, Yin, and Zhang~\cite{fengFastSamplingCounting2021}, is to bound the variance of the whole telescoping product directly.
Applied to a non-adaptive schedule of length $O(q\log h)$, this argument gives an $\varepsilon$-relative approximation using $O(q\varepsilon^{-2}\log h)$ independent Gibbs samples, with all query points fixed in advance.
Because the estimator naturally targets $Q^{-1}$, we first construct and analyze the schedule, then bound the estimator's variance, and finally take its reciprocal.

The schedule is the parameter-$1$ special case of the static schedule in~\cite[Algorithm 2]{harrisNearOptimalParallelApproximate2026}.
To keep the changes in $z$ between consecutive temperatures uniformly bounded, it uses two phases: fixed steps of size $1/h$ near $\beta_{\max}$, followed by geometrically growing steps.
It appends $\beta_{\min}$ and terminates when the next step reaches or passes $\beta_{\min}$ or when $d_i>2$.

\begingroup
\SetAlgoSkip{smallskip}
\begin{algorithm}[ht]
\caption{Construction of a non-adaptive cooling schedule~\cite[Algorithm 2]{harrisNearOptimalParallelApproximate2026}}
\label{alg:non-adaptive-cooling-schedule}
\KwIn{$\beta_{\min}$, $\beta_{\max}$, $q$ and $h$.}
\KwOut{A cooling schedule $B=(\beta_0,\ldots,\beta_\ell)$.}
Set $\beta_0 \gets \beta_{\max}$\;
\For{$i=0$ \textnormal{to} $+\infty$}{
  Set $d_i \gets \max\{1/h,(\beta_{\max}-\beta_i)/q\}$\;
  \If{$d_i>2$ \textnormal{or} $\beta_i-d_i\leq\beta_{\min}$}{
    \Return{$B=(\beta_0,\ldots,\beta_i,\beta_{\min})$}\;
  }
  Set $\beta_{i+1} \gets \beta_i-d_i$\;
}
\end{algorithm}
\endgroup

We call a step, or its corresponding interval, \emph{full} if it is generated by setting $\beta_{i+1}=\beta_i-d_i$ rather than by appending the endpoint $\beta_{\min}$.
Note that only the last interval may be non-full.

\begin{lemma}
    \label{lem:non-adaptive-cooling-schedule-length}
    The cooling schedule $(\beta_0,\ldots,\beta_\ell)$ returned by \pref{alg:non-adaptive-cooling-schedule} has length $\ell = O(q \log h)$.
\end{lemma}
\begin{proof}
    While $\beta_0-\beta_i\le q/h$, the procedure takes steps of size $1/h$, so this phase has at most $q+1$ full steps.
    Afterwards, every full step satisfies $\beta_0-\beta_{i+1}=(1+1/q)(\beta_0-\beta_i)$.
    Once $\beta_0-\beta_i>2q$, we have $d_i>2$, so the procedure stops.
    Thus the second phase has at most $\lceil\log_{1+1/q}(2h)\rceil+1=O(q\log h)$ steps.
    This proves the claimed length bound.
\end{proof}

Let $z(\beta)=\log Z(\beta)$ and for $i = 0, 1, \ldots, \ell - 1$, define $w_i=z(\beta_i)-z(\beta_{i+1})$.
\begin{lemma}
    \label{lem:non-adaptive-cooling-schedule-properties}
    The cooling schedule $(\beta_0,\ldots,\beta_\ell)$ returned by \pref{alg:non-adaptive-cooling-schedule} has the following properties:
    \begin{enumerate}
        \item $\sum_{i=0}^{\ell-1} w_i \le q$;
        \item $d_i \cdot z'(\beta_i) \le 1$ holds for all $i = 0, 1, \ldots, \ell-1$;
        \item $0 \le w_i \le 1$ holds for all $i = 0, 1, \ldots, \ell-1$.
    \end{enumerate}
\end{lemma}
\begin{proof}
    The first property follows by telescoping: $\sum_{i=0}^{\ell-1}w_i=z(\beta_0)-z(\beta_\ell)\le q$.
    For the second property, $z'(\beta)=\E[X\sim\mu_\beta]{H(X)}$ is nondecreasing since $z''(\beta)=\operatorname{Var}_{X\sim\mu_\beta}[H(X)]\ge0$.
    Hence $z'(\beta_i)\le z'(\beta_0)\le h$.
    For $i\ge1$, monotonicity also gives
    \begin{align*}
        q&\ge z(\beta_0)-z(\beta_i)
        =\int_{\beta_i}^{\beta_0}z'(\beta)\d\beta
        \ge(\beta_0-\beta_i) \cdot z'(\beta_i),
    \end{align*}
    so $z'(\beta_i)\le q/(\beta_0-\beta_i)$.
    For $i=0$, we have $d_0=1/h$, so $d_0 \cdot z'(\beta_0)\le1$. For $i\ge1$, taking the minimum of the two bounds gives $z'(\beta_i)\le\min\{h,q/(\beta_0-\beta_i)\}=1/d_i$, and hence $d_i \cdot z'(\beta_i)\le1$.

    Every full interval $i\le\ell-2$ has length $d_i$, so $w_i\le d_i \cdot z'(\beta_i)\le1$.
    The same bound applies to the last interval when its length is at most $d_{\ell-1}$.
    Otherwise, $d_{\ell-1}>2$, so $\E[X\sim\mu_{\beta_{\ell-1}}]{H(X)}\le1/d_{\ell-1}<1/2$.
    Since every positive value of $H$ is at least $1$, $\Pr[X\sim\mu_{\beta_{\ell-1}}]{H(X)=0}>1/2$, and hence $w_{\ell-1}<\log 2<1$.
    Nonnegativity follows from the monotonicity of $z$.
\end{proof}

Finally, we show that the internal intervals of the cooling schedule are mildly growing.
\begin{lemma}
    \label{lem:non-adaptive-cooling-schedule-mildly-growing}
    For any $i = 0, 1, \ldots, \ell-2$, we have $d_i \le d_{i + 1} \le (1 + 1/q) d_i$.
\end{lemma}
\begin{proof}
    Since step $i$ is full, $\beta_0-\beta_{i+1}=\beta_0-\beta_i+d_i$.
    Using the definition of $d_{i+1}$ and the fact that $(\beta_0-\beta_i)/q\le d_i$, we obtain $d_{i+1}=\max\{1/h,(\beta_0-\beta_i+d_i)/q\}\le(1+1/q)d_i$.
    The same maximum representation and $\beta_0-\beta_{i+1}\ge\beta_0-\beta_i$ give $d_{i+1}\ge d_i$.
\end{proof}

Given the above cooling schedule, we estimate the ratio using the classic product estimator.

\begingroup
\SetAlgoSkip{smallskip}
\begin{algorithm}[H]
\caption{Estimation of the partition-function ratio}
\label{alg:estimate-partition-ratio}
\KwIn{A cooling schedule $(\beta_0,\ldots,\beta_\ell)$ and an error parameter $\varepsilon\in(0,1)$}
\KwOut{An unbiased estimate $\overline Y$ of $Q^{-1}$}
Set $M \gets \left\lceil 4(\e^3-1)/\varepsilon^2 \right\rceil$\;
\For{$j$ \textnormal{from} $1$ \textnormal{to} $M$}{
  \For{$i$ \textnormal{from} $0$ \textnormal{to} $\ell-1$}{
    Sample $X_{i,j}\sim\mu_{\beta_i}$\;
    \eIf{$\beta_{i+1}>-\infty$}{
      Set $Y_{i,j} \gets \exp((\beta_{i+1}-\beta_i)\cdot H(X_{i,j}))$\;
    }{
      Set $Y_{i,j} \gets \mathbf{1}[H(X_{i,j})=0]$\;
    }
  }
  Set $Y_j \gets \prod_{i=0}^{\ell-1} Y_{i,j}$
}
Set $\overline Y \gets \frac{1}{M}\sum_{j=1}^M Y_j$\;
\Return{$\overline Y$}\;
\end{algorithm}
\endgroup

\begin{proof}[Proof of \pref{thm:main-result}]
    Run \pref{alg:estimate-partition-ratio} with error parameter $\varepsilon/2$ and return $\widehat Q=1/\overline Y$ when $\overline Y>0$, and $\widehat Q=0$ otherwise.
    For every $i$, including when $\beta_{i+1}=-\infty$, we have $\oE[Y_{i,j}]=Z(\beta_{i+1})/Z(\beta_i)$.
    Independence and telescoping give $\oE[Y_j]=\prod_{i=0}^{\ell-1}\oE[Y_{i,j}]=Q^{-1}$ and $\oE[\overline Y]=Q^{-1}$.
    For every full interval $i\le\ell-2$, direct calculation gives
    $$\log\frac{\oE[Y_{i,j}^2]}{\oE[Y_{i,j}]^2} = z(\beta_i-2d_i)+z(\beta_i)-2z(\beta_i-d_i) = w_i-\bigl(z(\beta_{i+1})-z(\beta_{i+1}-d_i)\bigr).$$
    For $i\le\ell-3$, the next interval is also full, so
    $\beta_{i+2}=\beta_{i+1}-d_{i+1}$.
    Since $z'$ is nondecreasing and $d_i\le d_{i+1}$, its average over
    the rightmost interval of length $d_i$ is at least its average over
    the interval of length $d_{i+1}$. Therefore,
    \begin{align*}
        \frac{z(\beta_{i+1})-z(\beta_{i+1}-d_i)}{d_i}
        \ge
        \frac{z(\beta_{i+1})-z(\beta_{i+1}-d_{i+1})}{d_{i+1}} 
        =\frac{w_{i+1}}{d_{i+1}},
    \end{align*}
    and hence
    $z(\beta_{i+1})-z(\beta_{i+1}-d_i)\ge\frac{d_i}{d_{i+1}}w_{i+1}.$
    
    \begin{itemize}
        \item If $\ell<2$, there are no full intervals.
        For $\ell\ge2$, dropping the nonnegative term for $i=\ell-2$ and applying \pref{lem:non-adaptive-cooling-schedule-properties} and \pref{lem:non-adaptive-cooling-schedule-mildly-growing} give
        $$\sum_{i=0}^{\ell-2} \log\frac{\oE[Y_{i,j}^2]}{\oE[Y_{i,j}]^2} \le w_0+\sum_{i=1}^{\ell-2} \left(1-\frac{d_{i-1}}{d_i}\right)w_i  \le w_0+\frac1{q+1}\sum_{i=1}^{\ell-2}w_i<2.$$

        \item If $\beta_\ell=-\infty$, $Y_{\ell-1,j}$ is an indicator and the logarithmic second-moment ratio is $w_{\ell-1}\le1$ by \pref{lem:non-adaptive-cooling-schedule-properties}.
        Otherwise, monotonicity of $z$ gives
        \begin{align*}
            \log\frac{\oE[Y_{\ell-1,j}^2]}{\oE[Y_{\ell-1,j}]^2}
            &=z(2\beta_\ell-\beta_{\ell-1})+z(\beta_{\ell-1})-2z(\beta_\ell)
            \le w_{\ell-1}\le1.
        \end{align*}
    \end{itemize}

    Independence of the $Y_{i,j}$ together with the above bounds gives
    $$\frac{\operatorname{Var}[Y_j]}{\oE[Y_j]^2}=\prod_{i=0}^{\ell-1}\frac{\oE[Y_{i,j}^2]}{\oE[Y_{i,j}]^2}-1=\exp\left(\sum_{i=0}^{\ell-1}\log\frac{\oE[Y_{i,j}^2]}{\oE[Y_{i,j}]^2}\right)-1<\e^{2+1}-1=\e^3-1.$$
    Chebyshev's inequality and $M=\lceil16(\e^3-1)/\varepsilon^2\rceil$ give $\Pr{|\overline Y/\oE[Y_j]-1|>\varepsilon/2}\le1/4$.
    Taking reciprocals on the complementary event gives $(1-\varepsilon)Q\le\widehat Q\le(1+\varepsilon)Q$.
    By \pref{lem:non-adaptive-cooling-schedule-length}, the algorithm uses $O(q\varepsilon^{-2}\log h)$ samples in one round.
\end{proof}

\section{Proofs of the Lower Bounds}

\subsection{Lower Bound for General Algorithms}

We prove the lower bound by constructing two instances with separated partition ratios but small KL divergence between their oracle outputs at every temperature. Any accurate estimator must distinguish the two instances, and the KL chain rule then gives the required number of queries, even when the queries are adaptive.
\generallowerbound*

\begin{proof}
    Fix $q\geq256$, $h\geq\e^{256}$, and $\varepsilon\in(0,1/4)$ throughout the proof.
    We construct a pair of instances $(\Gamma^0, \Gamma^1)$ with the same Hamiltonian values but different multiplicities, which any accurate estimator of $Q$ must distinguish.
    Specifically, let $\Gamma^0$ be a Gibbs instance with state space $\Omega$.
    Let $f:\Omega\to\mathbb Q_{\geq0}$ and $t\in\mathbb Q_{>0}$ be a function and a rational we choose later.
    Ideally, we want to construct $\Gamma^1$ so that for each $X\in\Omega$, there are $1+t \cdot f(X)$ copies of $X$ with the same Hamiltonian $H(X)$ in the state space of $\Gamma^1$.
    However, since $1+t\cdot f(X)$ may not be an integer, we multiply the number of states in $\Gamma^0$ and $\Gamma^1$ by a sufficiently large common multiplier so that all multiplicities are integers.
    Fix $-\infty\leq\beta_{\min}\leq\beta_{\max}<\infty$.
    For $b\in\{0,1\}$, let $Z^b(\beta)$ be the partition function of $\Gamma^b$ and set $Q^b=Z^b(\beta_{\max})/Z^b(\beta_{\min})$.   
    An immediate consequence of the construction is
    \begin{align*}
        \frac{Z^1(\beta)}{Z^0(\beta)}
        =\sum_{X\in\Omega}(1+t\cdot f(X)) \frac{\exp(\beta H(X))}{Z^0(\beta)} = \tp{1+t \cdot \mathbb E_{X\sim\mu^0_\beta}[f(X)]},
    \end{align*}
    Therefore,
    \begin{align}\label{eqn:general-ratio-identity}
        \frac{Q^1}{Q^0}
        =\frac{Z^1(\beta_{\max})/Z^0(\beta_{\max})}
        {Z^1(\beta_{\min})/Z^0(\beta_{\min})}
        =\frac{1+t \cdot \E[X\sim\mu^0_{\beta_{\max}}]{f(X)}}
        {1+t \cdot \E[X\sim\mu^0_{\beta_{\min}}]{f(X)}}.
    \end{align}
    
    For every $\beta\in[\beta_{\min},\beta_{\max}]$, let $\nu^b_\beta$ be the distribution of the oracle output $H(X)$ under $\Gamma^b$.
    The next claim bounds the KL divergence between the two oracle output distributions at any query temperature.
    Its proof is deferred to the end of this subsection.
    \begin{claim}
        \label{claim:general-query-kl}
        At every query temperature $\beta\in[\beta_{\min},\beta_{\max}]$, the oracle output distributions satisfy
        \[
            \DKL{\nu^1_\beta}{\nu^0_\beta}
            \leq t^2 \cdot \operatorname{Var}_{X\sim\mu^0_\beta}[f(X)].
        \]
    \end{claim}
    \noindent
    \pref{claim:general-query-kl} and \eqref{eqn:general-ratio-identity} suggest that we should choose $f$ with a large difference between its endpoint expectations and a uniformly small variance. With a suitable choice of $t$, this allows us to separate the partition ratios while keeping the oracle output distributions close at every temperature.


    We now specify $\Gamma^0$. Set $a=\lfloor\min\{q,\log h\}/16\rfloor$, $r=\lfloor q/(4a)\rfloor$, and $C=ra^2\cdot2^{r(a\lceil\e^{2a}\rceil+1)}$.
    For each tuple $(y_1,\ldots,y_r)\in\{0,\ldots,\lceil\e^{2a}\rceil\}^r$, we include $C(1-2^{-a})^{\sum_i y_i}$ states labeled by this tuple, each satisfying
    \[
        H(X)=1+\frac1r\sum_{i=1}^r y_i,
        \qquad
        f(X)=\sum_{i=1}^r\sum_{\ell=1}^{y_i}\frac1\ell.
    \]
    Also include one state $X$ with $H(X)=f(X)=0$.
    These multiplicities are positive integers because $C$ is divisible by $2^{ar\lceil\e^{2a}\rceil}$ and $\sum_i y_i\leq r\lceil\e^{2a}\rceil$.
    Set $t=16/(ra\lceil1/\varepsilon\rceil)$ and construct $\Gamma^1$ as above.

    Set $\beta_{\min}=-r\log(2(1-2^{-a}))$ and $\beta_{\max}=0$. Here $a\geq16$ and $r\geq4$.
    Both instances satisfy the bound on $h$, since every Hamiltonian value is at most $1+\lceil\e^{2a}\rceil$, we can verify that
    \[
    1 + \ceil{\e^{2a}} \leq 2 + \e^{2a} < \e^{16a} \leq h.
    \]
    For the bound on $q$, write the partition function as follows:

    \[
\begin{aligned}
    Z^0(\beta)-1
    &=
    \sum_{y_1=0}^{\ceil{\e^{2a}}}\cdots
    \sum_{y_r=0}^{\ceil{\e^{2a}}}
    C(1-2^{-a})^{\sum_{i=1}^r y_i}
    \exp\left(\beta+\frac{\beta}{r}\sum_{i=1}^r y_i\right)\\
    &=
    C\exp(\beta)\prod_{i=1}^r
    \left(
        \sum_{y_i=0}^{\ceil{\e^{2a}}}
        \bigl((1-2^{-a})\exp(\beta/r)\bigr)^{y_i}
    \right)\\
    &=
    C\exp(\beta)\left(
        \sum_{k=0}^{\ceil{\e^{2a}}}
        \bigl((1-2^{-a})\exp(\beta/r)\bigr)^k
    \right)^r.
\end{aligned}
\]
    The sum is at least $1$ at $\beta_{\min}$ and at most $2^a$ at $\beta_{\max}$. Including the state with Hamiltonian value $0$ decreases the ratio, so $\log Q^0<r(a+1)<q$.
    The bound $\sum_{\ell=1}^k1/\ell<3a$ for $1\leq k\leq\lceil\e^{2a}\rceil$ also gives $0\leq f(X)\leq3ra$.
    Since $3rat<12$, \eqref{eqn:general-ratio-identity} yields $\log Q^1<r(a+1)+\log13<q$.

    To bound the expectation and variance of $f(X)$ for $X\sim\mu^0_\beta$, fix any $\beta\in[\beta_{\min},\beta_{\max}]$ and set $z=(1-2^{-a})\exp(\beta/r)$, so that $1/2\leq z\leq1-2^{-a}$. The factorization above shows that, conditional on $H(X)>0$, the coordinates of the tuple labeling $X$ are independent, each with probability proportional to $z^k$ for $0\leq k\leq\lceil\e^{2a}\rceil$. Thus these coordinates have truncated geometric distributions. We first estimate the moments of the corresponding harmonic sums for independent, untruncated geometric random variables $Y_1,\ldots,Y_r$ with $\mathbb P[Y_i=k]=(1-z)z^k$ for integers $k\geq0$, and then account for truncation.
    Summing $\mathbb P[Y_1\geq\ell]=z^\ell$ and expanding the second moment give
    \[
        \E{\sum_{\ell=1}^{Y_1}\frac1\ell}=-\log(1-z),\qquad
        \operatorname{Var}\left[\sum_{\ell=1}^{Y_1}\frac1\ell\right]
        =\sum_{\ell=1}^{\infty}\frac{z^\ell}{\ell^2}<2.
    \]
    Since $1-z\geq\e^{-a}$, geometric summation gives $\mathbb P[Y_1>\lceil\e^{2a}\rceil]<1/8$ and the tail bound
    \[
        \E{Y_1 \cdot \mathbf{1}[Y_1>\lceil\e^{2a}\rceil]}
        \leq2\exp(3a-\e^a)<\frac1{16}.
    \]
    Conditioning on $Y_1\leq\lceil\e^{2a}\rceil$ therefore reduces the mean of $\sum_{\ell=1}^{Y_1}1/\ell$ by less than $1/16$, since this sum is increasing and at most $Y_1$.
    Bounding the squared deviation from the original mean gives conditional variance less than $16/7$.

    Conditioning each $Y_i$ on $Y_i\leq\lceil\e^{2a}\rceil$ preserves independence.
    By the factorization above, the resulting tuple $(Y_1,\ldots,Y_r)$ has the same distribution as the tuple labeling $X\sim\mu^0_\beta$ conditional on $H(X)>0$.
    Thus, given $H(X)>0$, $f(X)$ is a sum of independent coordinate functions, so its mean is at most $ra$ and its variance is less than $16r/7$.
    The states with Hamiltonian value $1$ have total weight at least $C\exp(\beta_{\min})\geq ra^2$, so $\mathbb P_{X\sim\mu^0_\beta}[H(X)=0]\leq1/(ra^2)$.
    Hence $0\leq\mathbb E_{X\sim\mu^0_\beta}[f(X)]\leq ra$, and accounting for $H(X)=0$ gives
    \[
        \operatorname{Var}_{X\sim\mu^0_\beta}[f(X)]
        <\frac{16r}{7}+\frac{(ra)^2}{ra^2}<4r.
    \]
    At $\beta=\beta_{\max}$, we have $z=1-2^{-a}$. Combining the truncated mean estimate with the bound on the probability of $H(X)=0$ shows that the expectation of $f(X)$ under $\mu^0_{\beta_{\max}}$ is greater than $r(a\log2-1/8)$.
    At $\beta=\beta_{\min}$, $z=1/2$ gives an untruncated mean of $r\log2$. Truncation can only decrease this mean, so the conditional expectation of $f(X)$ given $H(X)>0$ is at most $r\log2$. Since $f(X)=0$ when $H(X)=0$, the expectation of $f(X)$ under $\mu^0_{\beta_{\min}}$ is also at most $r\log2<r$.
    Thus the former expectation exceeds the latter by more than $ra/2$.
    Applying \eqref{eqn:general-ratio-identity}, we obtain
    \[
        \frac{Q^1}{Q^0}
        >1+\frac{8a}{a\ceil{1/\varepsilon}+16}>1+4\varepsilon
        >\frac{1+\varepsilon}{1-\varepsilon},
    \]
    where we used $a\lceil1/\varepsilon\rceil+16<2a/\varepsilon$.
    Combining the variance bound above with \pref{claim:general-query-kl}, we obtain $\mathcal D_{\mathrm{KL}}(\nu^1_\beta\|\nu^0_\beta)\leq1024\varepsilon^2/(ra^2)$ for every $\beta\in[\beta_{\min},\beta_{\max}]$.

    Let $\+A$ be an algorithm as in the theorem, and let $P^0,P^1$ be its output distributions on the two instances.
    The disjoint success intervals imply that $\+B=\{(1-\varepsilon)Q^1\leq\widehat Q\leq(1+\varepsilon)Q^1\}$ satisfies $P^1(\+B)\geq3/4$ and $P^0(\+B)\leq1/4$.
    Applying data processing to the indicator of $\+B$ gives
    \[
        \DKL{P^1}{P^0}
        \geq P^1(\+B)\log\frac{P^1(\+B)}{P^0(\+B)}
        +P^1(\overline{\+B})\log\frac{P^1(\overline{\+B})}{P^0(\overline{\+B})}
        \geq\frac{\log3}{2}>\frac12.
    \]
    Let $\+R$ denote all internal randomness of $\+A$, independent of the oracle's randomness.
    The public information and the law of $\+R$ are the same for both instances. Given $\+R$ and previous answers, the stopping decision and next query are therefore the same.
    Since the estimate is determined by $\+R$ and the answers, data processing and the KL chain rule give
    \[
        \frac12<\DKL{P^1}{P^0}
        \leq \E{\sum_{i\le N}\DKL{\nu^1_{\beta_i}}{\nu^0_{\beta_i}}}
        \leq\frac{1024N\varepsilon^2}{ra^2},
    \]
    where the expectation is over the algorithm's internal randomness and oracle answers during a run on the fixed instance $\Gamma^1$, and the sum runs over the queries actually made. 
    Each $\beta_i$ may depend on the internal randomness and previous answers.
    Since $ra^2=\Theta(q\min\{q,\log h\})$, this yields the required bound $N>ra^2/(2048\varepsilon^2)=\Omega(\min\{q\log h,q^2\}\cdot\varepsilon^{-2})$.
\end{proof}


\begin{proof}[Proof of~\pref{claim:general-query-kl}]
    Consider the distribution $\widetilde{\mu}^1_\beta$ by projecting copies of $X$ to $X$.
    Then, for any $X\in\Omega$, using \eqref{eqn:general-ratio-identity}, we have
    \begin{align}\label{eqn:R-f}
       R(X):= \frac{\widetilde{\mu}^1_\beta(X)}{\mu^0_\beta(X)} = \frac{1+t \cdot f(X)}{1+t \cdot \mathbb E_{X\sim\mu^0_\beta}[f(X)]}.
    \end{align}
    Note that $\mathbb{E}_{X\sim\mu^0_\beta} R(X)=1$.
    Then we have 
    \begin{align*}
    \DKL{\widetilde{\mu}^1_\beta}{\mu^0_\beta} &= \mathbb E_{X\sim\mu^0_\beta} [R(X)\log R(X)] \\
    &\le \mathbb E_{X\sim\mu^0_\beta} [R(X)^2]-1 
    = \operatorname{Var}_{X\sim\mu^0_\beta}[R(X)].
    \end{align*}
    By \eqref{eqn:R-f}, $R$ is an affine transformation of $f$, and therefore
    \begin{align*}
        \operatorname{Var}_{X\sim\mu^0_\beta}[R(X)] = \frac{t^2 \cdot \operatorname{Var}_{X\sim\mu^0_\beta}[f(X)]}
        {\bigl(1+t \cdot \E[X\sim\mu^0_\beta]{f(X)}\bigr)^2}.
    \end{align*}
    Applying the data processing inequality to $H(X)$, we obtain
    \[
        \DKL{\nu^1_\beta}{\nu^0_\beta} \le \DKL{\widetilde{\mu}^1_\beta}{\mu^0_\beta}
        \leq\frac{t^2 \cdot \operatorname{Var}_{X\sim\mu^0_\beta}[f(X)]}
        {\bigl(1+t \cdot \E[X\sim\mu^0_\beta]{f(X)}\bigr)^2}
        \leq t^2 \cdot \operatorname{Var}_{X\sim\mu^0_\beta}[f(X)],
    \]
    where the last inequality uses $f(X)\geq0$.
\end{proof}

\vspace{-0.5em}
\subsection{Lower Bound for Non-Adaptive Algorithms}

We next prove the stronger lower bound for non-adaptive algorithms.
We construct pairs of instances differing only in the number of states at one Hamiltonian value, which any accurate estimator must distinguish.
Choosing these values carefully bounds the total information each query provides across all pairs.
Since the queries are fixed before sampling, summing over the pairs yields the lower bound.
\nonadaptivelowerbound*

\begin{proof}
    Let $q \geq 256$, $h \geq \e^{256}$ and $\eps \in (0,1/4)$.
    Suppose that there exists a non-adaptive algorithm $\+A$ which makes $N$ oracle queries and succeeds with probability at least $3/4$ on every input.
    Throughout the proof, we fix $\beta_{\min} = -\infty$ and $\beta_{\max} = 0$.
    Define two parameters
    $$M = \floor{\frac{\log (h / 4)}{\log (1 + 1 / q)}}, \qquad r = \floor{\frac{\e^q}{8}}.$$
    Let $E_1,\ldots,E_M$ be $M$ distinct Hamiltonian values, where $E_1 = 1$ and $E_j = E_{j - 1} \cdot (1 + 1/q)$ for $2 \leq j \leq M$.
    By definition, $E_M \leq h / 4$.
    The Hamiltonian may also take the values $0$ and $rh / 4$.
    We then construct $M$ pairs of instances $(\Gamma^0_j, \Gamma^1_j)$ for $j \in [M]$, where the two instances in each pair only differ in the number of states at one Hamiltonian value $E_j$.
    To be more specific, let $N^0_j(x)$ and $N^1_j(x)$ be the number of states at Hamiltonian value $x$ in $\Gamma^0_j$ and $\Gamma^1_j$, respectively.
    For $j \in [M]$, we define the multiplicities of each Hamiltonian value as follows:
    \begin{itemize}
        \item If $x = 0$ or $x = rh / 4$, then $N^0_j(x) = N^1_j(x) = M \cdot \ceil{1 / \eps}$;
        \item If $x = E_i$ for some $i \neq j$, then $N^0_j(x) = N^1_j(x) = \ceil{1 / \eps}$;
        \item If $x = E_j$, then $N^0_j(x) = Mr \ceil{1 / \eps}$ and $N^1_j(x) = Mr \left(\ceil{1 / \eps} + 8\right)$.
    \end{itemize}
    Let $(Z^0_j(\beta), Z^1_j(\beta))$ be the partition functions of $(\Gamma^0_j, \Gamma^1_j)$ at parameter $\beta$ and let $(Q^0_j, Q^1_j)$ be the corresponding partition function ratios between $\beta_{\min} = -\infty$ and $\beta_{\max} = 0$.
    Moreover, let $\nu^0_{j, \beta}$ and $\nu^1_{j, \beta}$ be the distributions of the oracle output $H(X)$ under $\Gamma^0_j$ and $\Gamma^1_j$ at parameter $\beta$, respectively.

    To begin with, we verify that $q$ and $h$ are both valid input parameters for all instances.
    Note that for any $j \in [M]$, we have $Z^0_j(-\infty) = Z^1_j(-\infty)$ and $Z^0_j(0) < Z^1_j(0)$.
    Therefore,
    $$\log Q^0_j < \log Q^1_j = \log \tp{\frac{8r}{\ceil{1 / \eps}} + r + 3 - \frac{1}{M}} \le q.$$ 
    This finishes the verification for $q$.
    On the other hand, observe that all Hamiltonian values $E_j$ are at most $h / 4$, and hence contribute at most $h / 4$ to the expected Hamiltonian value. 
    The Hamiltonian value $rh / 4$ occurs with probability at most $1 / r$ under $\nu^0_{j, 0}$ and $\nu^1_{j, 0}$, and hence contributes at most $h / 4$ to the expected Hamiltonian value.
    This finishes the verification for $h$.

    We first claim the following uniform bound on the KL divergence between the oracle outputs of each instance pair.
    Assuming this bound, we finish the proof of \pref{thm:non-adaptive-lower-bound} by a standard information-theoretic argument.
    We defer the proof of the claim to the end of this section.
    \begin{claim}
        \label{claim:uniform-query-kl}
        Let $q\geq256$, $h\geq\e^{256}$, and $\eps\in(0,1/4)$, and let $M$, $r$, and $\nu^b_{j,\beta}$ be as in the construction above.
        For every query temperature $\beta\in(-\infty,0]\cup\{-\infty\}$,
        \[
            \sum_{j=1}^M\DKL{\nu^1_{j,\beta}}{\nu^0_{j,\beta}}
            \leq 2048\eps^2\left(1+\frac Mr\right).
        \]
    \end{claim}

    For every $j\in[M]$, note that $Z^0_j(-\infty) = Z^1_j(-\infty)$ and hence the partition ratios satisfy
    \[
        \frac{Q^1_j}{Q^0_j} = \frac{Z^1_j(0)}{Z^0_j(0)} = \frac{\frac{8r}{\ceil{1 / \eps}} + r + 3 - \frac{1}{M}}{r + 3 - \frac{1}{M}} > 1 + 3\eps > \frac{1 + \eps}{1 - \eps}.
    \]
    Let $P^0_j,P^1_j$ be the output distributions of $\+A$ on $\Gamma^0_j,\Gamma^1_j$.
    As in \pref{thm:general-lower-bound}, the disjoint success intervals imply $\mathcal D_{\mathrm{KL}}(P^1_j\|P^0_j)>1/2$ for every $j\in[M]$.
    
    Let $\+R$ denote all internal randomness of $\+A$, independent of the oracle's randomness.
    Since $\+A$ is non-adaptive and all instances have the same public information, fixing $\+R$ determines the same queries $\beta_i(\+R)$ for every instance.
    For every $b\in\{0,1\}$ and $j\in[M]$, conditional on $\+R$, the oracle outputs under $\Gamma^b_j$ have joint distribution $T^b_j(\+R)=\bigotimes_{i=1}^N\nu^b_{j,\beta_i(\+R)}$.
    The data processing inequality and the KL chain rule give, for every $j\in[M]$,
    \[
        \frac12<\DKL{P^1_j}{P^0_j}
        \leq \E[\+R]{\DKL{T^1_j(\+R)}{T^0_j(\+R)}}
        =\E[\+R]{\sum_{i=1}^N\DKL{\nu^1_{j,\beta_i(\+R)}}{\nu^0_{j,\beta_i(\+R)}}}.
    \]
    Summing over $j$ and applying \pref{claim:uniform-query-kl} at each query temperature gives
    \[
        \frac M2
        <\E[\+R]{\sum_{i=1}^N\sum_{j=1}^M
        \DKL{\nu^1_{j,\beta_i(\+R)}}{\nu^0_{j,\beta_i(\+R)}}}
        \leq2048\eps^2N\left(1+\frac Mr\right).
    \]
    Using $M=\Theta(q\log h)$ and $r=\Theta(\e^q)$, we conclude that
    \[
        N>\frac{Mr}{4096\eps^2(M+r)}
        \geq\frac{\min\{M,r\}}{8192\eps^2}
        =\Omega\left(\min\{q\log h,\e^q\} \cdot \eps^{-2}\right).
    \]
    This completes the proof of \pref{thm:non-adaptive-lower-bound}.
\end{proof}

\begin{proof}[Proof of~\pref{claim:uniform-query-kl}]
    At $\beta=-\infty$, both oracle outputs are identically zero, so their KL divergence is zero.
    Fix $\beta\in(-\infty,0]$.
    Note that \pref{claim:general-query-kl} does not rely on adaptivity.
    Taking $f(X)=\mathbf{1}[H(X)=E_j]$ and $t=8/\lceil1/\eps\rceil$ in \pref{claim:general-query-kl} gives the following bound for every $j\in[M]$:
    \[
        \DKL{\nu^1_{j,\beta}}{\nu^0_{j,\beta}}
        \leq64\eps^2 \cdot \nu^0_{j,\beta}(E_j)(1-\nu^0_{j,\beta}(E_j)).
    \]
    The probability of observing $E_j$ under the unperturbed instance is
    \[
        \nu^0_{j,\beta}(E_j)
        =\frac{r\e^{\beta E_j}}{r\e^{\beta E_j} + \tp{1+\e^{\beta rh/4}+\sum_{k\ne j}\e^{\beta E_k} / M}}\;.
    \]
    Since $\beta\le 0$, the denominator terms other than $r\e^{\beta E_j}$ sum to a value in $[1,3)$. 
    For $t > 0$, $t / (1 + t)^2 \leq \min\{t, 1 / t\} = \exp(-|\log t|)$.
    It implies that
    \[
        \nu^0_{j,\beta}(E_j)\bigl(1-\nu^0_{j,\beta}(E_j)\bigr)
        \leq\frac{3r\exp(\beta E_j)}{\bigl(1+r\exp(\beta E_j)\bigr)^2}
        \leq3\exp\bigl(-\lvert\beta E_j+\log r\rvert\bigr).
    \]
    At $\beta=0$, summing these bounds gives a total KL divergence of at most $192\eps^2M/r$.
    It remains to consider $\beta<0$.
    For $-\beta E_j\leq(\log r)/2$, the mean value theorem gives $\e^{-\beta E_j}/r\leq1/r+(-\beta E_j)/\sqrt r$.
    These indices form an initial segment, over which geometric summation bounds the sum of $-\beta E_j$ by $(q+1)\log r/2\leq q^2$.
    Since $r\geq\e^q/16$ and $q\geq256$,
    \[
        \sum_{j:-\beta E_j\leq(\log r)/2}\exp\bigl(-\lvert\beta E_j+\log r\rvert\bigr)
        \leq\frac Mr+\frac{q^2}{\sqrt r}<\frac Mr+1.
    \]
    For the remaining indices, consecutive values of $-\beta E_j$ differ by $-\beta E_j/q>(\log r)/(2q)>1/3$.
    Summing the geometric tails on either side of $\log r$ therefore gives
    \[
        \sum_{j:-\beta E_j>(\log r)/2}\exp\bigl(-\lvert\beta E_j+\log r\rvert\bigr)
        \leq2\sum_{k=0}^{\infty}\e^{-k/3}<8.
    \]
    Combining the bounds for the two ranges of indices proves the desired estimate:
    \[
        \sum_{j=1}^M\DKL{\nu^1_{j,\beta}}{\nu^0_{j,\beta}}
        <192\eps^2\left(\frac Mr+9\right)
        \leq2048\eps^2\left(1+\frac Mr\right).\qedhere
    \]
\end{proof}

\section*{AI Disclosure}
LLMs were used to refine the proofs and polish the language. 
The upper-bound arguments were developed by the authors, while the lower-bound constructions were developed with AI assistance. The authors take full responsibility for all content.

\section*{Acknowledgments}
We thank Weiming Feng for kindly allowing us to use his idea of analyzing the whole product.

\printbibliography

\end{document}